\documentclass[letterpaper,twocolumn,10pt]{article}
\usepackage[]{hyperref}         
\hypersetup{                    
  colorlinks,
  linkcolor={green!80!black},
  citecolor={red!70!black},
  urlcolor={blue!70!black}
}
\usepackage{usenix2019_v3}

\usepackage{tikz}
\usepackage{amsmath}
\usepackage[T1]{fontenc}
\usepackage[utf8]{inputenc}
\usepackage{enumitem}
\usepackage{amsmath,amssymb}
\usepackage{amsthm} 
\usepackage{graphicx}
\usepackage[toc,page]{appendix}
\usepackage{float}
\setlist[itemize]{label=$\cdot$}
\usepackage[english]{babel}
\usepackage{indentfirst}
\usepackage[english]{babel}
\usepackage{csquotes}
\usepackage{color, colortbl}
\definecolor{Gray}{gray}{0.9}
\definecolor{LightCyan}{rgb}{0.88,1,1}
\usepackage[first=0,last=9]{lcg}

\definecolor{junebud}{rgb}{0.74, 0.85, 0.34}
\newtheorem{theorem}{Theorem}[section]

\newtheorem{lemma}[theorem]{Lemma}
\usepackage{algorithm}
\usepackage{algorithmic}

\floatname{algorithm}{Procedure}

\usepackage{biblatex}
\newcount\Comments  
\newcommand{\kibitz}[2]{\ifnum\Comments=1{\color{#1}{#2}}\fi}

\makeatletter
\renewcommand*{\@fnsymbol}[1]{\ensuremath{\ifcase#1\or\dagger\else\@ctrerr\fi}}
\makeatother
\DeclareMathAlphabet{\mathcal}{OMS}{cmsy}{m}{n}

\begin{document}

\date{}

\title{\Large \bf Selfish Behavior in the Tezos Proof-of-Stake Protocol}

\author{
{\rm Michael Neuder}{\color{black}\thanks{Work completed while MN was visiting Harvard University.}}\\ 
{\small University of Colorado, Department of Computer Science \quad \quad} \\ 
{\small michael.neuder@colorado.edu} 
\and
{\rm Daniel J. Moroz}\\
{\small Harvard University, School of Engineering and Applied Sciences}\\
{\small dmoroz@g.harvard.edu}
\and
{\rm Rithvik Rao}\\
{\small Harvard University, School of Engineering and Applied Sciences}\\
{\small rithvikrao@college.harvard.edu}
\and
{\rm David C. Parkes}\\
{\small Harvard University, School of Engineering and Applied Sciences}\\
{\small parkes@eecs.harvard.edu}
} 

\maketitle

\begin{abstract}
Proof-of-Stake consensus protocols give rise to complex modeling challenges. We analyze the Babylon update (October 2019) to the Proof-of-Stake protocol on the Tezos blockchain, and demonstrate that, under certain conditions, rational participants are incentivized to behave dishonestly. In doing so, we provide a theoretical analysis of the feasibility and profitability of a block stealing attack that we call \textit{selfish endorsing}, a concrete instance of an attack previously only theoretically considered. We propose and analyze a simple change to the Tezos protocol which significantly reduces the (already small) profitability of this dishonest behavior, and introduce a new delay and reward scheme that is provably secure against length-1 and length-2 selfish endorsing attacks. Our framework provides a template for analyzing other Proof-of-Stake protocols for the possibility of selfish behavior.
\end{abstract}

\section{Introduction}
Blockchain technologies have received significant attention since the release of the Bitcoin protocol in 2008 \cite{nakamoto2008bitcoin}.
Blockchains arose as a solution to a critical problem in enabling permissionless, decentralized cryptocurrencies: how to maintain a global consensus on user account balances. The Bitcoin network addressed this problem by putting every transaction into a globally visible ledger protected by a Proof-of-Work (PoW) consensus protocol. In PoW, each block creator is tasked with assembling a list of valid transactions and doing costly computational work, and is rewarded with some amount of the native asset, Bitcoin. 

Proof-of-Stake (PoS) protocols are seen as potential successors of PoW protocols. Both kinds of protocols run a lottery to select the creator of the next block (the lottery is implicit in PoW, explicit in PoS), and in order to prevent malicious participants from creating many identities to increase their chances of winning, entry into these lotteries must be costly. PoW requires that lottery entrants burn computational cycles in order to join, while PoS requires participants to forego the use of staking capital for some duration of time. 
In particular, PoS protocols require that staked capital be forfeited if the  behavior  of a miner is not consistent  with that specified by the rules of the protocol.  While both approaches reward participants proportionally to expenditure, PoW has received criticism, and PoS is thought to provide an important step forward for consensus protocols on blockchains. 

The three main critiques of PoW protocols are the environmental impact, the inflationary tendencies, and the centralization found in many digital currencies.
The environmental argument is a result of PoW using significant amounts of energy; by some estimates, the annual energy consumption of the Bitcoin network is equivalent to that of Austria \cite{wastebitcoin}. The related inflationary criticism stems from the fact that it requires substantial real world expenditures by miners on hardware and electricity. These miners are in turn compensated by large block rewards that lead to  inflationary pressure on the currency, and if the rewards are reduced, fewer miners participate and the security of the network degrades. 
PoS protocols do not suffer from these limitations. 
Another  critique of PoW arises from concerns around the concentration of mining power. As of November 2019, F2Pool owned 18\% of the hashpower in the Bitcoin network, and the top four mining pools combined controlled more than 50\% of the hashpower \cite{pools}. However, because the ownership of cryptocurrencies is far from decentralized \cite{quantdecentralization},  PoS may not address this issue either.

  Tezos \cite{goodman2014tezos}, EOS \cite{eos}, Cardano (ADA) \cite{kiayias2017ouroboros}, BlackCoin \cite{vasin2014blackcoin}, and Nxt \cite{nxt} represent the major PoS networks that are currently running.\footnote{While many PoS protocols have been proposed, few are in actual use as of early 2020. PoS protocols have proven difficult to implement and pose novel technical challenges. As an example, Ethereum's PoS proposal, Casper \cite{buterin2017casper}, is still in development as of January 2020.} 
One of the distinguishing features of Tezos is that it has a built-in governance protocol, which allows changes to the specification to be voted on by participants in the P2P network.
In contrast, the development of Bitcoin has been slow because few developers want to risk forking the network over a protocol change \cite{lee20}. The design of the Tezos network hopes to encourage agreement on upgrades by creating a specific venue and timeline for voting on software updates. 
As of January 2020, each change required a quorum of participants and over 80\% approval to be instantiated \cite{postezos}.
On October 17, 2019, an update called {\em Babylon}~\cite{babylonupdate, babylondoc} was accepted into the Tezos protocol.

In the paper, we analyze a large component of this upgrade: a new consensus protocol called Emmy$^+$ \cite{emmyplusannounce}. In particular, we identify an incentive vulnerability in the Tezos PoS mechanism, analyze the severity of an attack, and establish security properties of alternative mechanisms. In doing so, we provide a simple theoretical framework that can  be used to analyze other implementations, and that  we hope will encourage a more formal treatment of the security properties of PoS protocols. 

Section~\ref{sec:2} describes how PoS is implemented in Tezos, Section~\ref{sec:3} formalizes the attack that we call \textit{selfish endorsing} and assesses the probability of it being profitable for a rational agent, Section~\ref{sec:4} proposes a simple heuristic fix to the Emmy$^+$ protocol, and Section~\ref{sec:5} proposes an alternative mechanism and proves that it is secure against a subset of selfish endorsing attacks, before concluding in Section~\ref{sec:6}.

\subsection{Related Work}

We seek to understand the extent to which rational participants in a particular PoS system can benefit by not behaving according to the protocol. This is analogous to the question asked by Eyal and Sirer (2013) \cite{eyal2018majority}, who demonstrate that miners can earn a higher proportion of rewards in a PoW protocol by deviating from the honest protocol and following a {\em selfish mining} strategy. Follow-up work includes Sapirshtein et al. (2016) \cite{sapirshtein2016optimal}, who identify the optimal such policy, Nayak et al. (2016) \cite{nayak2016stubborn}, who  consider network attacks, and Kwon et al. (2017) \cite{kwon2017selfish}, who consider the impact of selfish mining in the context of mining pools.

In regard to PoS, Brown-Cohen et al. (2019) show that complete security in their model of longest-chain PoS protocols is not possible \cite{brown2019formal}. The dishonest behavior that we refer to as \textit{selfish endorsing} is a real-world instance of the theoretical {\em predictable selfish mine} attack that appears in their work. 
We are not aware of any other academic work that formally analyzes the incentives of the Tezos PoS protocol. Nomadic Labs, the team that implemented Emmy$^+$ ~\cite{emmyplusannounce}, did provide a blog post with the results of an incentive analysis \cite{analysisemmyplus}, but without providing an explicit formalism for the model used or the probabilistic analysis.\footnote{We have verified with the authors of the blog post that the two models achieve similar numerical results when calculating the probability of a profitable attack using the same parameters.} In this work, we present the complete derivation of an attack model, and make explicit the methods used to obtain our results.

\section{Proof-of-Stake in Tezos}
\label{sec:2}

\subsection{The Basics}

Tezos implements an {\em optional Delegated Proof-of-Stake (DPoS) protocol} \cite{goodman2014tezos, postezos}, which is sometimes referred to as {\em Liquid Proof-of-Stake} \cite{lpos} to distinguish it from other,  more rigid DPoS implementations~\cite{eos}.\footnote{The system is described as ``optionally delegated" because owners of the currency can choose to delegate their baking and endorsing rights to another baker without forfeiting ownership of the token, but they are not compelled to as in a typical DPoS system such as EOS \cite{eos}. The term ``Liquid PoS" is used to emphasize this difference, and stems from the ideal of a liquid delegation market~\cite{lpos}.} 

Members of the Tezos consensus layer are called \textit{delegates} and are considered {\em active} when they participate in the creation and validation of blocks (and {\em passive} otherwise).
The Tezos unit of account (XTZ) is split into groups of 8,000 tokens called \textit{rolls}, and each delegate is associated with a set of rolls. Active delegates participate in a lottery to \textit{bake} and \textit{endorse} a block at every block-height in the chain. Bakers are responsible for including transactions in blocks while endorsers cryptographically sign the ``best''  block that they have seen at each height. 

The baking-and-endorsing priority lottery is carried out by randomly selecting rolls and giving the next available priority to the owner of that roll, a technique known as \textit{follow-the-Satoshi} \cite{bentov2014proof}. For each block-height, a list of bakers is created using the random roll selection process, and the index of a baker in this list determines the priority  with which they can create a block at this height. Additionally a set of 32 endorsers is created for each block-height. There is no priority list for endorsers, and  each has equal weight. Each draw from the set of rolls is done with replacement, so the same delegate may appear many times on the baking priority list as well as  in the set of endorsers.  

Bakers and endorsers are rewarded based on participation, which creates an incentive for delegates to remain active. Figure \ref{fig:endorsmentrewards} shows how the blocks are created and endorsed as well as the value of each, which is a function of the reward scheme described below. 

\subsection{The Babylon Upgrade and Emmy$^+$}

The new consensus protocol, Emmy$^+$, which was part of the October 2019 Babylon update, is distinct from its predecessor, Emmy, in three important ways:
\begin{enumerate} 
\item  {\bf A block's validity-time is a function of the number of endorsements it includes, in addition to the priority of the baker}. Each block includes a timestamp for when it was created. In order to regulate the rate at which blocks are created, Emmy$^+$ uses a minimum delay between blocks (see Equation \ref{delay}), which is what we refer to as the {\em validity-time}. When a node in the P2P network checks an incoming block, it ensures that the difference in the timestamp of that block and the previous block in the chain is greater than the validity-time. 

Before Emmy$^+$, the validity-time was only a function of the priority of the baker, but now it is also a function of the number of endorsements it includes. Importantly, the number of endorsements is not the number of delegates who endorse the block itself, but rather the number of endorsements for the previous block that this block includes. Since endorsements are  operations that are heard over the network,  including an endorsement in a block is analogous to including a transaction. 

In order for a block to be considered {\em valid}, its timestamp must differ from the previous block's timestamp by at least $\mathcal{D}$ seconds, where $\mathcal{D}$ defines the validity-time and is the following function of the baker's priority, $p$, and the number of endorsements included in the block, $e$ (see \textit{Minimal block delays} in \cite{postezos}):
\begin{align}
    \mathcal{D}(p, e) = 60 + 40  p + 8\max(24-e, \:0). \label{delay}
\end{align}

The effect is that for each priority-level that a baker is below the highest-priority (0), the validity-time for the block increases by 40 seconds. Also, for each endorsement that is missed below 24 of the 32 endorsers, the validity-time increases by 8 seconds; a block may miss up to 8 endorsements without incurring a time penalty, but each additional missed signature slows validity by 8 seconds.\footnote{Historical data for the Tezos network shows that the typical block misses only a few endorsements, and that these penalties have been quite rare~\cite{tzstats}.} If each block is baked by the $0^{th}$ priority delegate and endorsed by most of the endorser set, then a block will be created every 60 seconds and the Tezos network is considered healthy.

\item {\bf The fork-choice rule was changed}. The {\em fork-choice rule} is used by nodes in a P2P network to choose the branch that they will extend. Before this update, the {\em canonical fork} (i.e., the currently active fork) was the one  with the most endorsements (i.e., a heaviest-chain rule~\cite{goodman2014tezos}). After the update, the best fork became the one with the longest chain from the genesis block. This modified rule makes evaluation of branches easier, and alleviates a baker's uncertainty as to when to publish blocks in order to avoid missing out on late endorsements.\footnote{The uncertainty arises in the case where a baker has waited the whole validity-time, but hasn't heard many endorsements for the previous block. Now the baker is uncertain as to whether to publish now and risk a different block including more endorsements and thus becoming the active head, or to wait.}  Brown-Cohen et al.~\cite{brown2019formal} already showed that this longest-chain rule can  lead to theoretical vulnerabilities. Here we focus on one of these, {\em predictable selfish mining}, and our work represents a real-world instance of this vulnerability.

\item {\bf The rewards for baking and endorsing blocks were modified}. Before the change, baking a block earned the delegate a constant reward of 16 XTZ, but in Emmy$^+$ the block rewards are a function of the baker's priority, $p$, and the number of endorsements  $e$ for the previous block that are included in this block (see \textit{Rewards} in \cite{postezos}). Let $\mathcal{R}_{b}$ be the {\em baking reward}. In Emmy$^+$, we have
\begin{align}
    \mathcal{R}_{b}(p, e) &= \frac{16}{p+1} \left(\frac{4}{5} + \frac{1}{5} \cdot  \frac{e}{32}\right).
    \label{rewardblock}
\end{align}

The rewards for endorsements were also modified. Previously, endorsement rewards were a function of the priority of the block that the endorsement signed, but in Emmy$^+$ they are a function of the priority of the block that includes the endorsements. Denote the priority of the baker who baked the block that includes an endorsement as  $p_i$. Then the {\em endorsing reward}, $\mathcal{R}_{e}$, is calculated as 
\begin{align}
    \mathcal{R}_{e}(p_i) &= \frac{2}{p_i +1}.
    \label{rewardendorse}
\end{align}

Figure~\ref{fig:1} illustrates the rewards that the bakers and endorsers earn under the Emmy$^+$ rules.
\begin{figure*}
    \centering
    \includegraphics[scale=0.3]{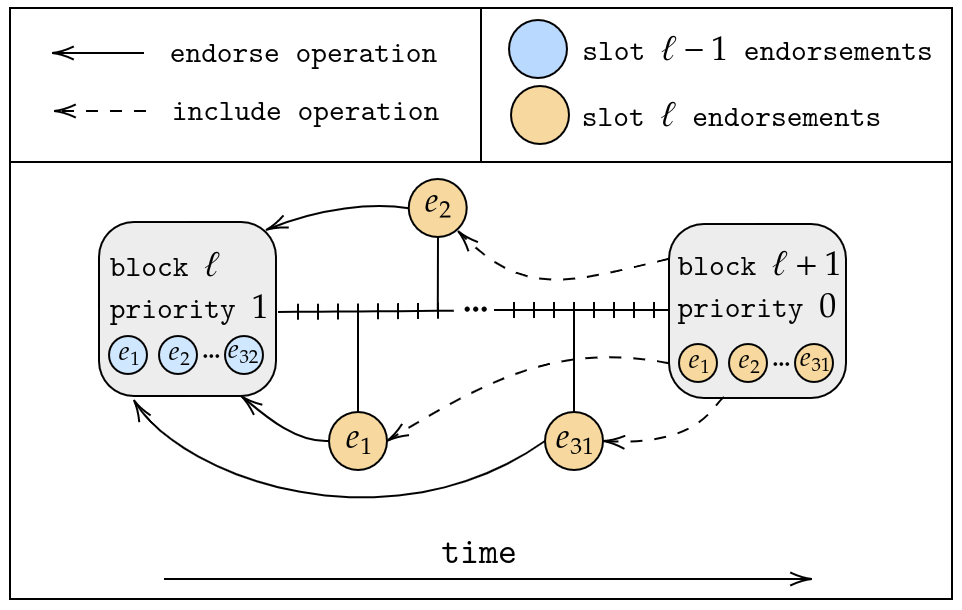}
    \caption{\label{fig:1} Let $\mathcal{R}_{b}(p,e)$ be as defined in \eqref{rewardblock} and let $\mathcal{R}_{e}(p)$ be as defined in \eqref{rewardendorse}. In this scenario, block $\ell$ will earn the baker a reward of $\mathcal{R}_{b}(1, 32) = 8$ XTZ, and block $\ell+1$ will earn the baker $\mathcal{R}_{b}(0, 31)=15.9$ XTZ. The endorsements for slot $\ell-1$ will each earn $\mathcal{R}_{e}(1) = 1$ XTZ and the endorsements for slot $\ell$ will each earn $\mathcal{R}_{e}(0) = 2$ XTZ. Notice that the slot $\ell$ endorsements still get the full reward when signing a lower priority block, simply because they are included in the $0^{th}$ priority block at the $\ell+1$ slot. Before the Emmy$^+$ upgrade the slot $\ell$ endorsements would only earn $1$ XTZ each. } 
    \label{fig:endorsmentrewards}
\end{figure*}
\end{enumerate}

\section{The Selfish Endorsing Attack}
\label{sec:3}

We now give an example of the  vulnerability that we call \emph{selfish endorsing}, which with some probability incentivizes a rational baker  to ignore the longest-chain rule and create a separate two-block fork faster than the rest of the network can publish two blocks. We give an example to illustrate that selfish endorsing is a profitable deviation from the intended protocol based on block and endorsement rewards alone. Additionally this attack provides an opportunity for a 1-confirmation double-spend.\footnote{A double-spend is an attack where the agent executes multiple transactions with the same tokens. The selfish endorsing attack makes this possible because the initial transaction can be included on an honest block, and then the attacker can override this block with a length-2 chain and not include the transaction, which allows the attacker to retain ownership of the tokens.} 

Let $\mathcal{X}$ be a rational delegate who is willing to deviate from the Emmy$^+$ protocol. 
We describe baking and endorsing rights in terms of {\em slots} that correspond to a specific length of the chain (e.g., a block baked at height $n$ from the genesis block occupies the $n^{th}$ slot). In the Emmy$^+$ implementation, delegates see exactly who will have baking and endorsing rights for the next several thousand ($ \approx 5 \times 4096$) blocks. 
For a given slot $\ell$, let $p_\ell$ be the highest priority and $e_\ell$ the number of endorsement rights that $\mathcal{X}$ is randomly allocated. Priorities are zero-indexed, with 0 being the highest.
Additionally, let $n_{\ell+1}$ denote the number of consecutive top baking priorities given to $\mathcal{X}$ at slot $\ell+1$ (e.g., if the baking priority list for the $\ell+1$ slot is $[\mathcal{X}, \mathcal{X}, \texttt{other}, \mathcal{X}, ...]$, then $p_{\ell+1} = 0 \text{ and } n_{\ell+1}=2$).

In our analysis, we make the assumption that messages are sent and received instantaneously. While it is out of scope handle propagation delays in this paper, this question has been studied in the context of Bitcoin~\cite{gobel2016delay} and offers an interesting avenue for future work.

\subsection{An Illustrative Attack on Emmy$^+$}

Figure~\ref{fig:2} shows that an attacker $\mathcal{X}$ with second priority ($p_\ell =1$) at slot $\ell$ can bake a block that will end up on the final chain despite the first-priority ($p_\ell =0$) delegate publishing on time.  This is malicious because the baker with the lower priority stole the block reward from the highest priority baker. 

\begin{figure*}
    \centering
    \includegraphics[scale=0.33]{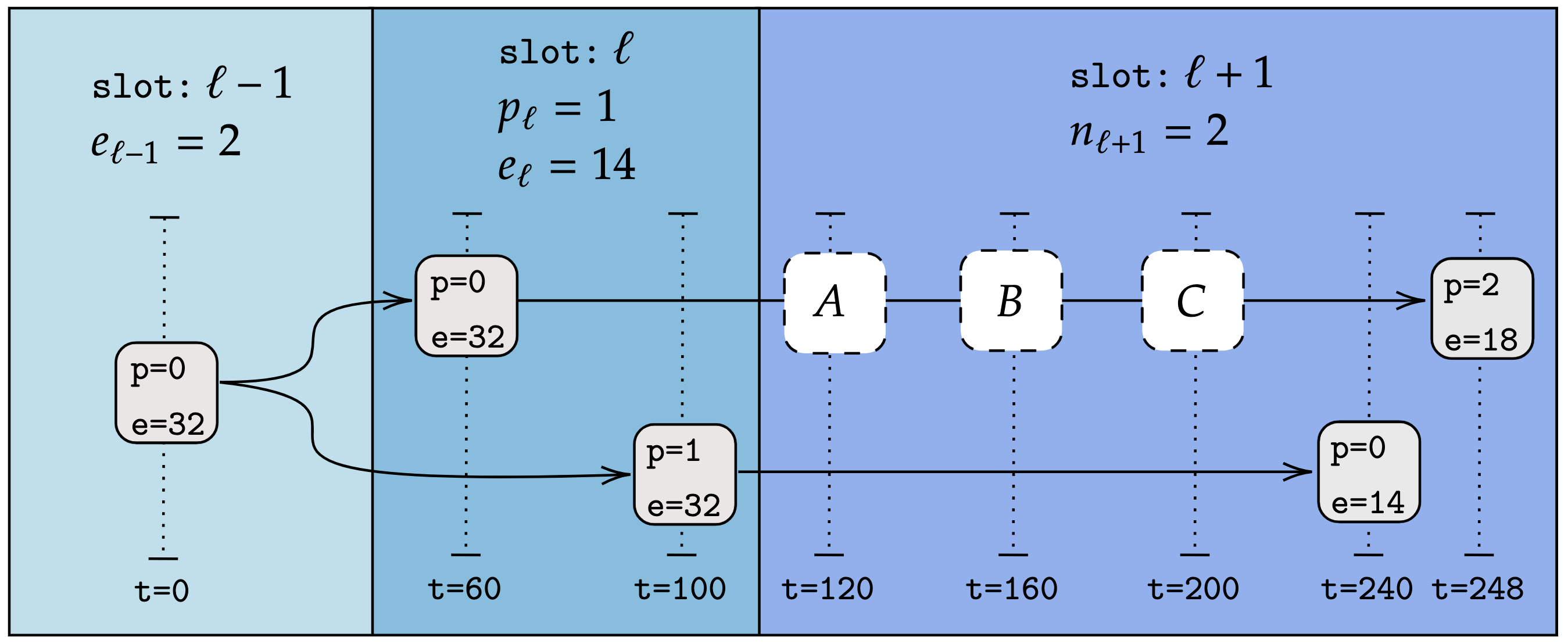}
    \caption{\label{fig:2} Let $p_\ell$ be the highest priority allocated to the attacker at slot $\ell$, and $n_{\ell+1}$ be the number of consecutive top priorities allocated to the attacker at slot $\ell+1$. Additionally, let $e_\ell$ be the number of endorsement rights that the attacker owns at slot $\ell$. This figure illustrates the selfish endorsing attack with $e_{\ell-1}=2,\: e_\ell=14, \: p_\ell=1, \text{ and } n_{\ell+1}=2$. The top and bottom forks show the next two blocks for the honest and selfish chains respectively. Each block shows the priority of the baker who created it, $p$, and the number of endorsements that it includes, $e$. Note that $e$ is not  the number of endorsements that this block receives, but rather the number of endorsements for the block baked in the previous slot that it includes, thus both blocks created at slot $\ell$ have $e=32$ because they both include all 32 endorsements of the block baked at slot $\ell-1$.  The empty blocks labelled $A, \, B$, and $C$ correspond to timestamps when the honest network would expect the second block to appear. $A$ and $B$ are empty because $\mathcal{X}$ holds the top two priorities for slot $\ell+1$. $C$ is empty because the baker with priority 2 for slot $\ell+1$ is missing $\mathcal{X}$'s 14 endorsements from slot $\ell$, which incurs a $48$ second penalty.}
\end{figure*}

Delegate $\mathcal{X}$ is able to look ahead and calculate when this attack can be executed  because the baking priorities and endorsing rights are publicly known into the future. We verify that the selfish chain will create two blocks faster than the honest chain using Equation \ref{delay}. 

Let $\mathcal{D}_h$ and $\mathcal{D}_s$ be the total time to create two blocks for the honest and selfish chains, respectively. We have,
\begin{align}
    \mathcal{D}_h &= \mathcal{D}(0, 32) + \mathcal{D}(2, 18) = 248 \\
    \mathcal{D}_s &= \mathcal{D}(1, 32) + \mathcal{D}(0, 14) = 240 
\end{align}

The combination of having a large share of endorsements for slot $\ell$ and the two highest priorities for slot $\ell+1$ allows $\mathcal{X}$ to slow down the honest network by only endorsing a private block at slot $\ell$. Attacker $\mathcal{X}$ can then produce a second valid block before the honest network does, and  create the unique longest chain. Using Equations~\ref{rewardblock} and~\ref{rewardendorse}, we can verify that this selfish behavior also results in a greater reward than following the honest protocol, thus demonstrating that this is a profitable deviation in its own right, and not just an opportunity for a 1-confirmation double spend. 

For this, let $\mathcal{R}_{h}$ denote the total reward earned by $\mathcal{X}$ over the next two blocks while behaving honestly. The labels underneath the expressions indicate the reason each reward is added.
We have
\begin{align}
    \mathcal{R}_{h} &= \underbrace{2  \mathcal{R}_{e}(0)}_{e_{\ell-1}=2} + \underbrace{14  \mathcal{R}_{e}(0)}_{e_{\ell}=14} + \underbrace{\mathcal{R}_{b} (0, 32)}_{p_{\ell+1}=0}  \notag \\ 
    &= 4 + 28 + 16 = 48.
\end{align}

Similarly we can calculate the total rewards that the attacker $\mathcal{X}$ earns while selfishly endorsing, denoted $\mathcal{R}_{s}$. We have
\begin{align}
    \mathcal{R}_{s} &= \underbrace{2  \mathcal{R}_{e}(1)}_{e_{\ell-1}=2} + \underbrace{14  \mathcal{R}_{e}(0)}_{e_{\ell}=14} + \underbrace{\mathcal{R}_{b}(1, 32)}_{p_\ell = 1} + \underbrace{\mathcal{R}_{b}(0, 14)}_{p_{\ell+1}=0 \And e_{\ell}=14} \notag \\
    &= 2 + 28 + 8 + 14.2= 52.2.
\end{align}

This demonstrates that the gain in reward from creating a new block at slot $\ell$ outweighs the loss in reward from the endorsements on slot $\ell-1$ ending up on a lower priority block, and 
because the block baked in slot $\ell+1$ only includes 14 endorsements. Critically, the rewards for the 14 endorsements for slot $\ell$ (which are essential in slowing down the honest network), do not decrease at all because these endorsements are still included on the block baked by $\mathcal{X}$ at slot $\ell+1$, which has priority $0$. 

Similarly to selfish mining, if this attack were to be executed it would be easily detected through the increased presence of orphan blocks, and the baker who was stealing the block would be identified, which may act as a deterrent from this behavior. However, it is worth noting that there is no mechanism within the Emmy$^+$ mechanism to punish this block stealing attack. The protocol as originally defined has the ability to slash a delegate's block and endorsement rewards if they are caught signing multiple blocks of the same height \cite{goodman2014tezos}, but selfish endorsing is different in that is impossible to cryptographically prove that a delegate is misbehaving and thus a punishment cannot be programmed in to the protocol.

\subsection{Feasibility}

Figure~\ref{fig:2} describes just one instance of a whole family of length-2 selfish endorsing attacks. In order to calculate the probability that any such attack is possible, we develop a generalized model. Consider tuples of the form $t = (e_{\ell-1}, \: e_\ell,\: p_\ell,\: n_{\ell+1})$. We define $t$ as \textit{feasible} if the selfish chain can create two valid blocks faster than the honest chain with this combination of parameters.  Figure~\ref{fig:2} demonstrates an attack that is feasible with the tuple $t_0 = (e_{\ell-1}=2,\: e_\ell=14, \: p_\ell=1, \: n_{\ell+1}=2)$, but we want to find all such combinations of parameters with this property. As above, we use $\mathcal{D}_h \text{ and } \mathcal{D}_s$  to denote the total time  to create two blocks, for the honest and selfish networks respectively.  Let $\mathcal{D}_2$ be the difference in the selfish and honest times, i.e., 
\begin{align}
    \mathcal{D}_2 = \mathcal{D}_s - \mathcal{D}_h.
\end{align} 

An attack is feasible if $\mathcal{D}_s < \mathcal{D}_h$, which implies $\mathcal{D}_2 < 0$.
\begin{lemma}
$\mathcal{D}_2(p_\ell, e_{\ell}, n_{\ell+1}) = 40  (p_\ell - n_{\ell+1}) + 8  \max(24 - e_\ell, 0) - 8  \max(e_\ell -8, 0).$ 
\end{lemma}
\begin{proof}
First, we can express $\mathcal{D}_h$ in terms of $e_\ell, p_\ell$, and $n_{\ell+1}$ using Equation \ref{delay}, as 
\begin{align}
    \mathcal{D}_h(p_\ell, e_\ell, n_{\ell+1}) ={}& \underbrace{60}_{\mathcal{D}(0, 32)} +  \underbrace{60 + 40 n_{\ell+1}}_{\mathcal{D}(n_{\ell+1}, 32-e_\ell)}  \notag \\ 
    \quad{}&+ \underbrace{8 \max(24 - (32-e_\ell), 0)}_{\mathcal{D}(n_{\ell+1}, 32-e_\ell)} \notag \\ 
    ={}& 120 + 40  n_{\ell+1} + 8 \max(e_\ell - 8, 0).
\end{align}

Similarly, we can express $\mathcal{D}_s$ as
\begin{align}
    \mathcal{D}_s(p_\ell, e_\ell, n_{\ell+1}) &= \underbrace{60 + 40 p_\ell}_{\mathcal{D}(p_\ell, 32)} + \underbrace{60 + 8 \max(24-e_\ell, 0)}_{\mathcal{D}(0, e_\ell)} \notag \\ 
    &= 120 + 40  p_\ell + 8  \max(24 - e_\ell, 0).
\end{align}

Now we solve for $\mathcal{D}_2$: 
\begin{align}
    \mathcal{D}_2(p_\ell, e_\ell, n_{\ell+1}) ={}& \mathcal{D}_s(p_\ell, e_\ell, n_{\ell+1}) - \mathcal{D}_h(p_\ell, e_\ell, n_{\ell+1}) \notag \\ 
    ={}& 120 + 40  p_\ell + 8  \max(24 - e_\ell, 0) \notag \\
    \quad{}&- \left[120 + 40  n_{\ell+1} + 8  \max(e_\ell - 8, 0) \right] \notag \\ 
    ={}& 40  (p_\ell - n_{\ell+1}) + 8  \max(24 - e_\ell, 0) \notag \\
    \quad{}&- 8  \max(e_\ell -8, 0).
\end{align}
\end{proof}

Now if $\mathcal{D}_2<0$,  then the attack is feasible. Notice that $\mathcal{D}_2$ is not a function of $e_{\ell-1}$ because both chains will always include all the endorsement operations for the $n-1$ slot.

\subsection{Profitability} 

Now we need to  parameterize the reward functions. As above, we use $\mathcal{R}_{h}$ 
to denote the reward for delegate $\mathcal{X}$ playing honestly for the next two blocks, and $\mathcal{R}_{s}$ to denote the reward for selfish endorsing over that same span. Further, let $\mathcal{R}_{2}$ denote the difference in selfish and honest rewards, i.e.,
\begin{align}
    \mathcal{R}_{2} &= \mathcal{R}_{s} - \mathcal{R}_{h}.
\end{align}

An attack is \textit{profitable} if the amount of reward that $\mathcal{X}$ receives playing selfishly is greater than that which they would receive playing honestly, or if $\mathcal{R}_{2} > 0$. We are slightly abusing notation here, in that for delay, the subscripts $s \text{ and } h$ refer to the delay for $\mathcal{X}$ versus the delay for the rest of the network. In the case of rewards, the subscripts both refer to delegate $\mathcal{X}$ and correspond to selfish or honest behavior. 

We express $\mathcal{R}_{2}$ as a function of the tuple $t =  (e_{\ell-1}, e_\ell, p_\ell, n_{\ell+1})$.
In this case, $\mathcal{R}_{2}$ does not depend  on $n_{\ell+1}$.
\begin{lemma}\label{lem42}
$ \mathcal{R}_{2}(p_\ell, e_{\ell-1}, e_\ell) = 16 \!\left(\frac{1}{p_\ell +1} + \frac{e_\ell}{160} - \frac{1}{5}   \right) + 2e_{\ell-1} \!\left( \frac{1}{ p_\ell+1} - 1\right)$.
\end{lemma}
\begin{proof}
First we focus on the reward for behaving honestly. We must take into account the endorsement rewards and the block rewards for the next two blocks. Using Equations \ref{rewardblock} and \ref{rewardendorse}, we define the total reward for honest behavior over the next two blocks as,
\begin{align}
    \mathcal{R}_{h}(e_{\ell-1}, e_\ell)  &= \underbrace{e_{\ell-1} \cdot \mathcal{R}_{e}(0)}_{\texttt{slot } \ell \texttt{ rewards}} + \underbrace{e_\ell \cdot \mathcal{R}_{e}(0) + \mathcal{R}_{b} (0, 32)}_{\texttt{slot } \ell +1\texttt{ rewards}} \notag \\
    &= 2e_{\ell-1} + 2 e_\ell + 16 = 2 (e_{\ell-1} + e_\ell) + 16.
\end{align}

We categorize the rewards for the endorsements under the next slot because that is where they are included in a block. Similarly, we calculate the total reward for selfish behavior over the next two blocks as,
\begin{align}
    \mathcal{R}_{s}(&e_{\ell-1}, e_\ell, p_\ell) ={} \underbrace{e_{\ell-1} \cdot \mathcal{R}_{e}(p_\ell) + \mathcal{R}_{b}(p_\ell, 32)}_{\texttt{slot } \ell \texttt{ rewards}} 
  \notag   \\ & \quad \quad \quad  \quad \quad \quad  + \underbrace{e_\ell \cdot \mathcal{R}_{e}(0) + \mathcal{R}_{b}(0, e_\ell)}_{\texttt{slot } \ell +1 \texttt{ rewards}} \notag \\
    ={}& \frac{2}{p_\ell +1}  e_{\ell-1} + \frac{16}{p_\ell+1} 
     + 2 e_\ell + 16  \left(\frac{4}{5} + \frac{1}{5} \cdot \frac{e_\ell}{32}\right) \notag \\
    ={}& 2  \left( \frac{e_{\ell-1}}{p_\ell +1} + e_\ell \right) 
     + 16  \left(\frac{1}{p_\ell+ 1} + \frac{4}{5} + \frac{e_\ell}{160} \right).
\end{align}

Now we solve for $\mathcal{R}_{2}$, obtaining
\begin{align}
&    \mathcal{R}_{2} (e_{\ell-1}, e_\ell, p_\ell) ={} \mathcal{R}_{s}(e_{\ell-1}, e_\ell, p_\ell) - \mathcal{R}_{h}(e_{\ell-1}, e_\ell) \notag \\ 
    ={}& 2  \left( \frac{e_{\ell-1}}{p_\ell +1} + e_\ell - e_{\ell-1} - e_\ell \right) \notag
    \\ \quad{}&+ 16  \left(\frac{1}{p_\ell+ 1} + \frac{4}{5} + \frac{e_\ell}{160} - 1\right) \notag \\
    ={}& 2e_{\ell-1}  \left( \frac{1}{p_\ell + 1} - 1\right)
    + 16  \left(\frac{1}{p_\ell+1} + \frac{e_\ell}{160} - \frac{1}{5} \right).
\end{align}
\end{proof}
Now if $\mathcal{R}_{2} > 0$, then an attack parameterized is profitable.

\subsection{Probability}

As a third step, we can calculate the probability of a tuple $t = (e_{\ell-1}, e_\ell, p_\ell, n_{\ell+1})$ occurring randomly on the chain. Let $\alpha$ denote the fraction of active rolls that $\mathcal{X}$ owns. Given this, the probability of $\mathcal{X}$ receiving any priority or endorsement for slot $\ell$ is $\alpha$. Let $\mathcal{P} \sim \text{Geometric}(\alpha)$ be the random variable representing the number of consecutive slots not owned by $\mathcal{X}$.\footnote{Recall the intuition for a geometric random variable as counting the number of failures until the first success, where success has probability $\alpha$. In this case a failure is defined as not being allocated a priority, which happens with probability $1-\alpha$.} The distribution on $\mathcal{X}$ being allocated each of the first $n$ priorities in slot $\ell+1$ is also a geometric random variable, but with probability $(1-\alpha)$, and denoted $\mathcal{N} \sim \text{Geometric}(1-\alpha)$. Lastly, the distribution on $\mathcal{X}$  being allocated $e_\ell$ endorsement rights for slot $\ell$ is a binomial random variable with 32 draws, denoted $\mathcal{E} \sim \text{Binomial}(32, \alpha)$. Thus, we can  calculate the probability of tuple $t= (e_{\ell-1}, e_\ell, p_\ell, n_{\ell+1})$ given $\alpha$, as
\begin{align}
&    \mathrm{Pr} [ \, t \: | \: \alpha \,] ={} \underbrace{(1 - \alpha)^{p_\ell}\alpha}_{\mathrm{Pr}[\mathcal{P} = p_\ell]} \times \underbrace{\alpha^{n_{\ell+1}} (1 - \alpha)}_{\mathrm{Pr}[\mathcal{N} = n_{\ell+1}]} \notag
    \\ \quad{}&\times \underbrace{ \binom{32}{e_{\ell-1}} \alpha^{e_{\ell-1}} (1-\alpha)^{32-e_{\ell-1}}}_{\mathrm{Pr}[\mathcal{E} = e_{\ell-1}]} 
    \times\underbrace{\binom{32}{e_{\ell}} \alpha^{e_{\ell}}(1-\alpha)^{32-e_{\ell}}}_{\mathrm{Pr}[\mathcal{E} = e_\ell]}\notag \\
    ={}& \binom{32}{e_{\ell-1}} \cdot \binom{32}{e_{\ell}} \cdot  \alpha^{n_{\ell+1} + e_{\ell-1} + e_\ell + 1}\cdot(1 - \alpha)^{65 +p_\ell  - e_{\ell-1} - e_{\ell}}.
\end{align}

\subsection{Generalizing}

We can now  calculate the probability of this family of length-2 attacks occurring. Let  $\mathcal{A}_2$ denote all tuples for which the attack is feasible (i.e.,  it creates two blocks faster than the honest network) and profitable (i.e., it has higher rewards than playing honestly). We have,
\begin{align}
\mathcal{A}_2 &= 
\{(e_{\ell-1}, e_\ell, p_\ell, n_{\ell+1}) \: | \: \mathcal{D}_2 < 0 \: \land \: \mathcal{R}_{2} > 0 \}.
\end{align}

We also want to measure how profitable these attacks are. Let $\mathcal{V}_2$ be the expected increase in reward of the attacks in $\mathcal{A}_2$ (i.e., how much more profit the selfish behavior 
will generate than by playing honestly). We have,
\begin{align}
    \mathcal{V}_2 &= \sum_{t \in \mathcal{A}_2} \mathrm{Pr} [ \, t \: | \: \alpha \,]  (\mathcal{R}_{s} (t) - \mathcal{R}_{h}(t)).
\end{align}

Procedure 1 (in Appendix \ref{appb}) demonstrates  these calculations, and the blue columns in Table \ref{tab:tab1} show the results of the analysis. Let $\mathcal{C} = 365 \cdot 24 \cdot 60$ represent the number of minutes in a year, so $\mathcal{C} \cdot \mathrm{Pr}[\mathcal{A}_2]$ is the expected number of attacks per year and $\mathcal{C} \cdot \mathcal{V}_2$ is the expected increase in value (in XTZ) for following the selfish policy for a year. This shows the attack is not a serious threat, given that even with 40\% of the stake, $\mathcal{X}$ is only expected to earn 254.94 XTC ($\approx$ \$307.23 in November 2019) more than if they had played honestly for the year. Regardless, this is an example of how this type of attack could be formulated against a  longest-chain PoS system.

\begin{table*}
    \centering
\begin{tabular}{|c||cc|c||cc|c|}
\hline
$\alpha$ & \multicolumn{2}{c|}{$\mathcal{C} \cdot \mathrm{Pr}[\mathcal{A}_2]$} & \% & \multicolumn{2}{c|}{$\mathcal{C} \cdot \mathcal{V}_2 $} & \% \\ \hline \hline
0.1 &\cellcolor{LightCyan} 0.04 & \cellcolor{junebud}0.17 & 425\% & \cellcolor{LightCyan}0.09 & \cellcolor{junebud}0.21 & 233\% \\
0.15 &\cellcolor{LightCyan} 3.88 & \cellcolor{junebud}2.16 & 56\% & \cellcolor{LightCyan}7.07 & \cellcolor{junebud}2.02 & 29\% \\
0.2 &\cellcolor{LightCyan} 33.91 & \cellcolor{junebud}7.70 & 23\% & \cellcolor{LightCyan}52.61 & \cellcolor{junebud}6.10 & 12\% \\
0.25 &\cellcolor{LightCyan} 136.76 & \cellcolor{junebud}12.91 & 9.4\% & 1\cellcolor{LightCyan}75.91 & \cellcolor{junebud}9.00 & 5.1\% \\
0.3 &\cellcolor{LightCyan} 309.66 & \cellcolor{junebud}12.66 & 4.1\% & \cellcolor{LightCyan}324.55 & \cellcolor{junebud}7.92 & 2.4\% \\
0.35 &\cellcolor{LightCyan} 407.33 & \cellcolor{junebud}8.07 & 2.0\% & \cellcolor{LightCyan}361.14 & \cellcolor{junebud}4.60 & 1.3\% \\
0.4 &\cellcolor{LightCyan} 318.98 & \cellcolor{junebud}3.53 & 1.1\% & \cellcolor{LightCyan}254.94 & \cellcolor{junebud}1.85 & 0.7\% \\ \hline
\end{tabular}
    \caption{Let $\alpha$ be the percentage of rolls that the attacker owns, and let $\mathcal{C} = 365 \cdot 24 \cdot 60$ be the number of minutes in a year. Then $\mathcal{C} \cdot \mathrm{Pr}[\mathcal{A}_2]$ represents the number of times the attack will be feasible and profitable over the course of a year, and $\mathcal{C} \cdot \mathcal{V}_2 $ will be the additional profit over honest behavior that the attacker makes in a year. The light blue cells represent the results under the current implementation and the green cells represent the results after our heuristic fix (described in Section~\ref{sec:4}) is applied. The \% column shows the respective ratio of the green column to the blue column (e.g., 1\% means that after the fix, the value of the attack is 1\% of what it is currently). A subtle feature of these results is that the probability of an attack being profitable and feasible is not monotone increasing with respect to $\alpha$. This is due to the fact that the endorsements owned by the attacker for the $\ell-1$ slot will end up on a lower priority block and thus be worth less in the case of an attack. As a result, the more endorsement slots owned by the attacker at $\ell-1$, the less profit they stand to gain by stealing the block. We see that there is a balance between having a large enough $\alpha$ to make the attack feasible, but not too large as to have a high probability of having many endorsement rights. We numerically calculated that an attacker makes the most profit when $\alpha=0.351$.}
    \label{tab:tab1}
\end{table*}

\section{A Heuristic Fix to Emmy$^+$}
\label{sec:4}

The profit from the attack can be reduced by introducing a simple modification to the Emmy$^+$ consensus protocol. Recall that in Emmy$^+$ the endorsement reward defined in Equation \ref{rewardendorse} is a function of the priority of the block that includes the endorsement, denoted $p_i$. If the rewards for endorsements are switched to being a function of the block that they endorse, denoted $p_e$, the attacks occur less frequently.

The green columns in Table \ref{tab:tab1} represent the probability and value of the attack at different levels of $\alpha$ after this fix, and the \% column reports the improvement over the status quo. Recall that $\alpha$ represents the fraction of active rolls that are owned by $\mathcal{X}$. The only exception to this reduction is the case of $\alpha=0.1$, where we observe that both the probability and the value of the attacks rise as a result of the suggested change. Since the expected increase in value of 0.21 XTZ $\approx \$0.25$ over a year is effectively no increase in profit, the fix still seems reasonable. See the caption of Table \ref{tab:tab1} for more discussion.

The values of the green columns in Table \ref{tab:tab1} were calculated using Procedure 1 in Appendix \ref{appb}, but with the reward function presented in Lemma~\ref{lemmareduced}. 
We defer the proof of this lemma to the Appendix  because it is similar to Lemma~\ref{lem42}.

\begin{lemma}
\label{lemmareduced}
If the endorsement rewards $\mathcal{R}_{e}$ are instead a function of the block that they endorse, $p_e$, then $ \mathcal{R}_{2}(p_\ell, e_{\ell-1}, e_\ell) = 16 \cdot \left(\frac{1}{p_\ell +1} + \frac{e_\ell}{160} - \frac{1}{5}   \right) + 2e_{\ell} \left( \frac{1}{ p_\ell+1} - 1\right)$.
\end{lemma}
\begin{proof}
See Appendix \ref{appa}.
\end{proof}

This is not a security proof, but rather a heuristic change that decreases the probability and profitability of selfish endorsing for most values of $\alpha$. 

\section{Alternative Fix: New Delay and Reward Functions}
\label{sec:5}

In this section, we present modifications to both the delay and reward functions that would  
make the PoS protocol provably secure against a specific attack vector (in this case, length-1 and length-2 selfish endorsing).  In contrast to the heuristic fix in Section~\ref{sec:4}, we  change the validity-time and reward functions, as opposed to simply switching which block is used to calculate endorsement rewards.\footnote{The modified   reward scheme was proposed by Arthur Breitman during discussions about  potential tweaks to the Emmy$^+$ protocol in response to an earlier version of this paper.}
 In practice, creating a fully secure PoS protocol is not as simple as implementing these functions, because there are other long-range forking attacks that must be taken into account. If the delay for each drop in baker priority is too large, then the honest network is slowed dramatically, which allows for an attacker to have more time to create a longer private chain.

Let $\mathcal{D}'$ be the new delay function for a block being valid, which is a function of the priority of the baker, $p$, and the number of endorsements it includes, $e$. We define this as,
\begin{align}
\label{eq:newdelay}
    \mathcal{D}'(p, e) &= 60 + 193 p + 8  \max(24-e,0).
\end{align}

Relative to Emmy$^+$, the  new component is the amount of time added for each drop in priority of the baker. In Emmy$^+$ this value is 40, while in this revised expression $\mathcal{D}'$ it is 193. This delay function was constructed to ensure that the highest priority block can \textit{always} be published before the second highest priority (see Lemma \ref{lemmadouble}). This is because given the following reward functions, length-1 block stealing attacks are profitable in some cases (they are not profitable in Emmy$^+$ because of Lemma \ref{lemmasingle}), so in order to insure they do not occur, the delay function ensures that they are never feasible. 

Let $\mathcal{R}_{b}'$ be the modified reward function for a block baked with priority $p$, and including $e$ endorsements. It maintains an 80 XTZ per block inflation rate, but splits the rewards 40/40 between the baker and the endorsers:
\begin{align}\label{newrewardblock}
    \mathcal{R}_{b}'(p, e) =  \frac{5}{4} \cdot \frac{e}{p+1}
\end{align}

Let $\mathcal{R}_{e}'(p_i)$ denote the reward for an endorsement that is included in a block baked with priority $p_i$. We have,

\begin{align}\label{newrewardendorse}
    \mathcal{R}_{e}'(p_i) = \frac{5}{4} \cdot \frac{1}{p_i+1}
\end{align}

We see that if a block is baked with priority 0 and all 32 endorsements are included from the previous block, then the total reward for the block becomes $\mathcal{R}_{b}'(0,32) + 32 \mathcal{R}_{e}'(0) = 40 + 40 = 80$ XTZ. 
\medskip

In Section~\ref{sec:51}, we prove that length-1 selfish endorsing attacks are not feasible under this new delay schedule. In Section~\ref{sec:52}, we 
prove that length-2 selfish endorsing attacks are not profitable under this new reward protocol.

\subsection{Security Against Length-1 Attacks}
\label{sec:51}

Lemma~\ref{lemmasingle} establishes a useful result. It explains why we have 
not needed to consider single-block selfish endorsing attacks until this point.
\begin{lemma}\label{lemmasingle}
For any tuple $(e_{\ell-1}, p_\ell)$, the length-1 selfish endorsing attack is not profitable under Emmy$^+$. 
\end{lemma}
\begin{proof}
See Appendix \ref{appc}.
\end{proof}

What about the new delay schedule~\eqref{eq:newdelay}? In fact, we show that a length-1 selfish endorsing attack is never even feasible. Assuming the honest network has the highest priority baking rights, let $\mathcal{D}'_{h,1}$ be the time for the honest network to create a single block under the new delay function, and $\mathcal{D}'_{s,1}$ be the time for the selfish delegate to create a single block.
\begin{lemma}\label{lemmadouble} 
Under the modified PoS protocol, for any tuple $(e_{\ell-1}, p_\ell)$, we have $\mathcal{D}'_{h,1}(p_\ell, e_{\ell-1}) < \mathcal{D}'_{s,1}(p_\ell, e_{\ell-1})$.
\end{lemma}
\begin{proof}
In the worst case, the honest network will not receive any endorsements, so the slowest the block creation could be is,
\begin{align}
    \mathcal{D}'_{h,1}(p_\ell, e_{\ell-1}) &\leq \mathcal{D}'_{h,1}(0, 0) = 60 + 8 \cdot 24  
    \leq 252.
\end{align}

In the best case for the attacker, they will own all 32 endorsements, and a priority of 1 ($2^{nd}$ best) in the block, and 
\begin{align}
    \mathcal{D}'_{s,1}(p_\ell, e_{\ell-1}) &\geq \mathcal{D}'_{s,1}(1, 32) = 60 + 193 =     253.
\end{align}

Putting this together, we have
\begin{align}
    \mathcal{D}'_{h,1} \leq 252 < 253 \leq \mathcal{D}'_{s,1}\  \Rightarrow \ \mathcal{D}'_{h,1} < \mathcal{D}'_{s,1}.
\end{align}
\end{proof}

\subsection{Security Against Length-2 Attacks}
\label{sec:52}

We now turn to length-2 selfish endorsing attacks, and will establish that they are never profitable under the modified protocol. We first need a closed form representation of the rewards $\mathcal{X}$ would receive  under the new reward function, denoted $\mathcal{R}'_{h}$ and  $\mathcal{R}'_{s}$ for playing  honestly and selfishly, respectively. The  derivation of Lemma~\ref{lemmanewrew} is  similar to Lemma~\ref{lem42}, and  deferred to Appendix \ref{appd}.
\begin{lemma}\label{lemmanewrew}
Under the modified PoS protocol, the total reward for a rational delegate $\mathcal{X}$ playing honestly over the next two blocks with the tuple $(e_{\ell-1}, e_\ell, p_\ell)$ is
\begin{align}
     \mathcal{R}'_{h}(e_{\ell-1}, e_\ell)  &= 1.25  (e_{\ell-1} + e_\ell + 32),
\end{align}
and the total reward for $\mathcal{X}$ to play selfishly over the next two blocks  is
\begin{align}
    \mathcal{R}'_{s}(e_{\ell-1}, e_\ell, p_\ell) &= 2.5 \left( \frac{e_{\ell-1}}{p_\ell+1} + e_\ell\right).
\end{align}
 \end{lemma}
\begin{proof}
See Appendix \ref{appd}.
\end{proof}

We now prove that length-2 selfish endorsing under this reward system is never profitable. 
\begin{lemma}\label{lemma:len2noprofit}
Under the modified PoS protocol, for any tuples $(e_{\ell-1}, e_\ell, p_\ell)$, we have  $\mathcal{R}_{s}'(e_{\ell-1}, e_\ell, p_\ell) \leq \mathcal{R}_{h}'(e_{\ell-1}, e_\ell)$.
\end{lemma}
\begin{proof}
Assume for contradiction that $\mathcal{R}_{s}' > \mathcal{R}_{h}'$, and so
\begin{align}
    2.5 \left( \frac{e_{\ell-1}}{p_\ell+1} + e_\ell\right) &> 1.25  (e_{\ell-1} + e_\ell + 32).
\end{align}

Collecting terms and simplifying, this is equivalent to
\begin{align}
\label{eq:simple}
e_{\ell-1} \left( \frac{1 - p_\ell}{p_\ell+1}\right) + e_\ell &> 32.
\end{align}

Because $p_\ell \geq 1$ , we have
\begin{align}
    \frac{1 - p_\ell}{p_\ell+1} \leq 0.
\end{align}

This along with the fact that $e_{\ell-1} \geq 0$ implies
\begin{align*}
    e_{\ell-1} \left( \frac{1 - p_\ell}{p_\ell+1}\right) \leq 0.
\end{align*}

Substituting this into~\eqref{eq:simple}  implies that we would need  $e_\ell > 32$ for profitability, and a contradiction, since $e_\ell \in \{0, 1, ..., 32\}$. 
\end{proof}

\subsection{Discussion}

This is a useful result in the context of defending against length-1 and length-2 selfish endorsing. And yet, this modification actually weakens the PoS system against longer forking attacks. This is because under the modified delay schedule, the time penalty for each drop in priority is so large that an attacker with a couple first priorities in the coming slots could withhold their blocks, and during the intermediate time work on creating a longer chain, while the honest nodes have to wait 193 extra seconds each time the attack has the $0^{th}$ priority. The last section of the analysis of Emmy$^+$ provided by Nomadic Labs discusses these kinds of trade-offs, and how the exact constants were selected for Emmy$^+$ \cite{analysisemmyplus}.

\section{Conclusion}
\label{sec:6}

We have demonstrated that live PoS systems can be formally analyzed for their incentive vulnerabilities. Our analysis of Tezos also serves as a real-world example of the {\em predictable selfish mining attack} that was theorized by Brown-Cohen et al.~\cite{brown2019formal}.
At the same time, we recognize that, as of November 2019, length-2 selfish endorsing attacks do not seem to be a major threat to the Tezos network because the profit from executing such an attack is so low. However, we have also suggested a simple, heuristic fix that would reduce the probability and value of many of the attacks by an order of magnitude. We have further presented a modified delay schedule and reward function that is provably secure against length-1 and length-2 selfish endorsing (though we acknowledge that in practice other attack vectors must also be considered when implementing these protocols).

By no means have we provided a complete security analysis of the Tezos protocol. Future work combining selfish endorsing with other deviating strategies, as was done in the PoW literature \cite{nayak2016stubborn, kwon2017selfish}, is critical to assessing the total security of the system. The formal framework introduced in our work  can be used to check other PoS protocols for potential vulnerabilities to selfish behavior by parameterizing a model of time and reward with respect to a specific protocol. In fact, a new reward scheme was recently proposed by the Nomadic Labs developers under the  name  {\em Carthage}~\cite{cathagedoc}, and this framework can be applied directly to the new delay and reward functions. Nomadic Labs provide an analysis of the effect of this reward scheme on a few different off-protocol strategies, including the selfish endorsing attack \cite{cathageanal}.

We hope that our work serves as a starting point for analyses and a more formal treatment of the security properties of PoS systems. For immediate  future directions, the first is to develop a theory of profitable selfish-endorsing attacks beyond length-2. This is non-trivial because the technique described in Appendix \ref{appb} scales exponentially due to the Cartesian product of each of the state variables. The second direction is to consider a more general forking attack that on its own would earn a smaller staking reward relative to honest behavior, but allow for an attacker to include a double-spend transaction.\footnote{Both of these questions have been considered by Nomadic Labs~\cite{analysisemmyplus}, but the precise models, derivations, and probabilistic machinery used are not made explicit.} Additionally, more analysis needs to be done in regard to modelling the presence of multiple selfish delegates, which also remains an active area of research in the Bitcoin PoW literature~\cite{multi1, multi2}. 

\section*{Acknowledgments} The authors would like to thank Eugen Zalinescu and Arthur Breitman for helpful discussions and for proposing new reward schemes to analyze. Additional thanks are due to Dan Robinson and the two anonymous reviewers for helpful comments as we prepared the manuscript for the 2020 Cryptoeconomic Systems Conference.  This work is supported in part by two generous gifts to the Center for Research on Computation and Society at Harvard University, both to support research on applied cryptography and society. Daniel J. Moroz was also supported in part by the Ethereum Foundation.

\printbibliography

\appendix
\section{Proof of Lemma \ref{lemmareduced}}\label{appa}
\begin{lemma}
$ \mathcal{R}_{2}(p_\ell, e_{\ell-1}, e_\ell) = 16 \left(\frac{1}{p_\ell +1} + \frac{e_\ell}{160} - \frac{1}{5}   \right) + 2e_{\ell} \left( \frac{1}{ p_\ell+1} - 1\right)$.
\end{lemma}
\begin{proof}
First we focus on the reward for behaving honestly. We must take into account the endorsement rewards and the block rewards for the next two blocks. We define the total reward for honest behavior over the next two blocks as
\begin{align}
    \mathcal{R}_{h}(e_{\ell-1}, e_\ell)  &= \underbrace{e_{\ell-1} \cdot \mathcal{R}_{e}(0)}_{\texttt{slot } \ell -1 \texttt{ rewards}} + \underbrace{e_\ell \cdot \mathcal{R}_{e}(0) }_{\texttt{slot } \ell \texttt{ rewards}}+ \underbrace{\mathcal{R}_{b} (0, 32)}_{\texttt{slot } \ell+1 \texttt{ rewards}} \notag \\
    &= 2e_{\ell-1} + 2 e_\ell + 16 = 2  (e_{\ell-1} + e_\ell) + 16.
\end{align} 

Similarly we calculate the rewards while following the selfish endorsing policy  as 

\begin{align}
    \mathcal{R}_{s}(&e_{\ell-1}, e_\ell, p_\ell) ={} \underbrace{e_{\ell-1} \cdot \mathcal{R}_{e}(0)}_{\texttt{slot } \ell-1 \texttt{ rewards}}+\underbrace{ \mathcal{R}_{b}(p_\ell, 32)+e_\ell \cdot \mathcal{R}_{e}(p_\ell) }_{\texttt{slot } \ell \texttt{ rewards}} 
   \notag  \\ & + \underbrace{ \mathcal{R}_{b}(0, e_\ell)}_{\texttt{slot } \ell +1 \texttt{ rewards}} \notag \\
    ={}& 2  e_{\ell-1} + \frac{16}{p_\ell+1} + \frac{2}{p_\ell+1} e_\ell +16 \left(\frac{4}{5} + \frac{1}{5} \cdot \frac{e_\ell}{32}\right) \notag \\
    ={}& 2  \left( \frac{e_{\ell}}{p_\ell +1} + e_{\ell-1} \right) + 16  \left(\frac{1}{p_\ell+ 1} + \frac{4}{5} + \frac{e_\ell}{160} \right).
\end{align}

Now we can solve for $\mathcal{R}_{2}$:
\begin{align}
    \mathcal{R}_{2} &(e_{\ell-1}, e_\ell, p_\ell) ={} \mathcal{R}_{s}(e_{\ell-1}, e_\ell, p_\ell) - \mathcal{R}_{h}(e_{\ell-1}, e_\ell) \ \notag \\ 
    ={}& 2 \left( \frac{e_{\ell}}{p_\ell +1} + e_{\ell-1} - e_{\ell-1} - e_\ell \right) + 16  \left(\frac{1}{p_\ell+ 1} + \frac{4}{5} + \frac{e_\ell}{160} - 1\right) \notag \\
    ={}& 2e_{\ell}  \left( \frac{1}{p_\ell + 1} - 1\right)  + 16  \left(\frac{1}{p_\ell+1} + \frac{e_\ell}{160} - \frac{1}{5} \right).
\end{align}
\end{proof}

\section{Procedure 1}\label{appb}
This procedure is used to calculate the total probability and value of a length-2 selfish endorsing attack. Note that $\times$ is the Cartesian Product of  lists.
\begin{algorithm}[H]
\caption{Find attack probability $\&$ value}
\begin{algorithmic}
\REQUIRE $\alpha$
\STATE \texttt{E $\leftarrow$ [0, 1, ..., 32]} 
\STATE \texttt{P $\leftarrow$ [1, 2, ..., 20]} 
\STATE \texttt{N $\leftarrow$ [1, 2, ..., 20]} 
\STATE \texttt{totalProb} $\leftarrow 0$
\STATE \texttt{totalValue} $\leftarrow 0$
\FOR{$(e_{\ell-1}, e_\ell, p_\ell, n_{\ell+1}) \in \texttt{E} \times \texttt{E} \times \texttt{P} \times \texttt{N}$}
\IF{$\mathcal{D}_2 < 0\text{ and } \mathcal{R}_{2} > 0 $}
\STATE \texttt{currentProb} $\leftarrow \mathrm{Pr}[\,(e_{\ell-1}, e_\ell, p_\ell, n_{\ell+1}) \: | \: \alpha \,]$
\STATE \texttt{totalProb += currentProb}
\STATE \texttt{totalValue += currentProb $*\:\: \mathcal{R}_{2}$}
\ENDIF
\ENDFOR
\RETURN{\texttt{(totalProb, totalValue)}}
\end{algorithmic}
\end{algorithm}

\section{Proof of Lemma \ref{lemmasingle}} \label{appc}

\begin{lemma}
For any tuple $(e_{\ell-1}, p_\ell)$, the length-1 selfish endorsing attack is not profitable under Emmy$^+$. 
\end{lemma}
\begin{proof}
First we consider the time it takes the honest network to produce a single block at height $\ell$. Denote this value $\mathcal{D}_{h,1}$.
We have,
\begin{align}
    \mathcal{D}_{h,1}(e_{\ell-1}) &= 60 + 8 \max(24-(32-e_{\ell-1}), 0)\notag \\
    &= 60 + 8  \max(e_{\ell-1}-8, 0).
\end{align}

Now we find the time it takes for the selfish delegate to create a block, $\mathcal{D}_{s,1}$.
We have,
\begin{align}
    \mathcal{D}_{s,1}(p_{\ell}, e_{\ell-1}) &= 60 + 40p_{\ell} + \max(24-e_{\ell-1}, 0).
\end{align}

The best case scenario for the attacker is when $p_\ell=1$, and for this we have
\begin{align}
    \mathcal{D}_{s,1}(e_{\ell-1}) &= 100 + \max(24-e_{\ell-1}, 0).
\end{align}

Now we find the number of endorsements required for the selfish delegate to produce a valid block faster than the honest network to be $e_{\ell-1}=19$. We verify this with the following calculations:
\begin{align}
    \mathcal{D}_{s,1}(19) &= 100 + 8 (5)= 140 \\ 
    \mathcal{D}_{h,1}(19) &= 60 + 8 (11)= 148
\end{align}

Now  we know that $e_{\ell-1}= 19$ is the best case scenario for the attack being feasible. If $e_{\ell-1} < 19$, then the selfish network will not be able to create a block fast enough, and if $e_{\ell-1} > 19$, then the reward will be lower because additional endorsements will end up on the lower priority stolen block. 
Additionally, we know that the reward for playing honestly is $\mathcal{R}_{h,1}(e_{\ell-1}) = 2e_{\ell-1}$ because we will get all the endorsement rewards, and the reward for playing selfishly is,
\begin{align}
    \mathcal{R}_{s,1}(e_{\ell-1}) &= \underbrace{8  \left(\frac{4}{5} + \frac{1}{5} \cdot \frac{e_{\ell-1}}{32}\right)}_{\mathcal{R}_{b}(1,e_{\ell-1})} + e_{\ell-1} \cdot \underbrace{1}_{\mathcal{R}_{e}(1)}
\end{align}

Plugging in $e_{\ell-1}=19$, we have
\begin{align}
    \mathcal{R}_{h,1}(19) &= 19 ( 2)  = 38 \\
    \mathcal{R}_{s,1}(19) &= 8 \left(\frac{4}{5} + \frac{1}{5} \cdot \frac{19}{32}\right) + 19 = 26.35
\end{align}

We see that even for the smallest value of $e_{\ell-1}$ that makes the attack feasible, the profit gained from creating a new block does not outweigh the profit lost as a result of the endorsements ending up on a worse block. So because $\mathcal{R}_{h,1} > \mathcal{R}_{s,1}$,  a single block selfish endorsing attack is not profitable under Emmy$^+$. 
\end{proof}

\section{Proof of Lemma \ref{lemmanewrew}} \label{appd}
\begin{lemma}
Under the modified PoS protocol, the total reward for a rational delegate $\mathcal{X}$ playing honestly over the next two blocks with the tuple $(e_{\ell-1}, e_\ell, p_\ell)$ is
\begin{align*}
     \mathcal{R}'_{h}(e_{\ell-1}, e_\ell)  &= 1.25   (e_{\ell-1} + e_\ell + 32),
\end{align*}
and the total reward for $\mathcal{X}$ to play selfishly over the next two blocks is
\begin{align*}
    \mathcal{R}'_{s}(e_{\ell-1}, e_\ell, p_\ell) &= 2.5 \left( \frac{e_{\ell-1}}{p_\ell+1} + e_\ell\right).
\end{align*}
\end{lemma}
\begin{proof}
This derivation is  similar to the proof of Lemma~\ref{lem42}, but substituting
for the new reward functions. We have:
\begin{align}
    \mathcal{R}'_{h}(&e_{\ell-1}, e_\ell)  ={} \underbrace{e_{\ell-1} \cdot \mathcal{R}'_{e}(0)}_{\texttt{slot } \ell \texttt{ rewards}} + \underbrace{e_\ell \cdot \mathcal{R}'_{e}(0) + \mathcal{R}'_{b} (0, 32)}_{\texttt{slot } \ell +1\texttt{rewards}} \notag \\
    ={}& 1.25 e_{\ell-1} +  1.25 e_\ell + 1.25(32) = 1.25   (e_{\ell-1} + e_\ell + 32). 
    \end{align}
    \begin{align}
    \mathcal{R}'_{s}(&e_{\ell-1}, e_\ell, p_\ell) = \underbrace{e_{\ell-1} \cdot \mathcal{R}'_{e}(p_\ell) + \mathcal{R}'_{b}(p_\ell, 32)}_{\texttt{slot } \ell \texttt{ rewards}} \notag
    \\ \quad{}&\quad + \underbrace{e_\ell \cdot \mathcal{R}'_{e}(0) + \mathcal{R}'_{b}(0, e_\ell)}_{\texttt{slot } \ell +1 \texttt{ rewards}} \notag \\
    ={}& \frac{1.25}{p_\ell +1}  e_{\ell-1} + \frac{1.25 e_{\ell-1}}{p_\ell+1} + 1.25 e_\ell + 1.25e_\ell \notag \\
    ={}& \frac{2.5e_{\ell-1}}{p_\ell+1} + 2.5e_\ell  \notag \\ ={}& 2.5 \left( \frac{e_{\ell-1}}{p_\ell+1} + e_\ell\right).
\end{align}
\end{proof}

\end{document}